\documentclass{llncs}

\usepackage{amsmath}
\usepackage{amssymb}
\usepackage{todonotes}
\usepackage[colorlinks=true, linkcolor=blue, urlcolor=blue, citecolor=blue]{hyperref}
\usepackage[capitalise]{cleveref}

\usepackage{tikz}
\usetikzlibrary{arrows,calc,topaths}
\tikzstyle{small}=[font=\footnotesize]
\tikzset{
    every picture/.style={>=stealth,auto,node distance=5cm},
    label/.style={font=\scriptsize},
    mgnode/.style={font=\boldmath},
}
\usepackage{graphicx}
\usepackage{thmtools}
\usepackage{thm-restate}

\usepackage[utf8]{inputenc}
\usepackage[
            backend=bibtex,
            doi=false,
            url=false,
            isbn=false,
            backref=false,
            style=numeric,
            firstinits=true,
            maxcitenames=5, 
            maxnames=5,
            ]{biblatex}
\DeclareFieldFormat[report]{title}{\mkbibquote{#1\isdot}}

\usepackage{wrapfig}

\usepackage{centernot}
\usepackage{color}
\usepackage{tikz}
\usepackage{yfonts}

\newcommand{\Act}{\mathit{Act}}
\newcommand{\ignore}[1]{}

\newcommand{\ol}[1]{\overline{#1}}

\newcommand{\fbound}{c}

\newcommand{\MCOL}[1]{#1}

\newcommand{\tinc}{\MCOL{trace inclusion}}
\newcommand{\Tinc}{\MCOL{Trace inclusion}}
\newcommand{\tincp}{\MCOL{trace inclusion problem}}

\newcommand{\tuniv}{\MCOL{trace universality}}
\newcommand{\Tuniv}{\MCOL{Trace universality}}

\newcommand{\OCA}{\MCOL{OCA}}

\newcommand{\textOCAs}{\MCOL{one-counter automata}}

\newcommand{\TextOCAs}{\MCOL{One-counter automata}}
\newcommand{\DOCA}{\MCOL{DOCA}}

\newcommand{\OCN}{\MCOL{OCN}}
\newcommand{\OCNs}{\MCOL{OCNs}}
\newcommand{\textOCN}{\MCOL{one-counter net}}
\newcommand{\textOCNs}{\MCOL{one-counter nets}}

\newcommand{\TextOCNs}{\MCOL{One-counter nets}}
\newcommand{\DOCN}{\MCOL{DOCN}}

\newcommand{\textDOCNs}{\MCOL{deterministic one-counter nets}}

\newcommand{\CM}{\MCOL{CM}}
\newcommand{\textCM}{\MCOL{counter machine}}
\newcommand{\textCMs}{\MCOL{counter machines}}

\newcommand{\TextCMs}{\MCOL{Counter machines}}

\newcommand{\textDPDAs}{\MCOL{deterministic pushdown automata}}

\newcommand{\incCM}{\MCOL{ICM}}
\newcommand{\textIncCM}{\MCOL{incrementing counter machine}}

\newcommand{\N}{\mathbb{N}}

\newcommand{\Z}{\mathbb{Z}}

\newcommand{\x}{\times}
\renewcommand{\O}{\mathcal{O}}

\newcommand{\Xstep}[4]{\ensuremath{\,{\stackrel{#1}{#4}}\!{}^{\scriptstyle{#2}}_{\scriptstyle{#3}}}\,}
\newcommand{\notXstep}[4]{\ensuremath{\,{\centernot{\stackrel{#1}{#4}}\!{}}^{\scriptstyle{#2}}_{\scriptstyle{#3}}}\,}
\newcommand{\Step}[3]{\Xstep{#1}{#2}{#3}{\longrightarrow}} 
\newcommand{\notStep}[3]{\notXstep{#1}{#2}{#3}{\longrightarrow}}
\newcommand{\step}[1]{\Step{#1}{}{}}

\newcommand{\WStep}[3]{\Xstep{#1}{#2}{#3}{\Longrightarrow}}

\newcommand{\wstep}[1]{\WStep{#1}{}{}}

\newcommand{\norm}[1]{|#1|}
\newcommand{\norminf}[1]{\norm{#1}_{\infty}}
\newcommand{\bottom}{\ensuremath{\bot}}

\newcommand{\PSPACE}{PSPACE}

\newcommand{\NN}[1]{{\ensuremath{\cal #1}}}
\newcommand{\AN}[1]{{\ensuremath{\cal #1}}}

\newcommand{\effect}[1]{\Delta_{#1}}

\newcommand{\guard}[1]{\Gamma_{#1}}

\newcommand{\SP}{\NN{A}}
\newcommand{\DUP}{\NN{B}}
\newcommand{\prefix}[2]{{}^{#1}#2}
\newcommand{\deffect}[1]{\effect{\DUP}(#1)}
\newcommand{\seffect}[1]{\effect{\SP}(#1)}
\newcommand{\dguard}[1]{\guard{}'(#1)}
\newcommand{\sguard}[1]{\guard{}(#1)}

\newcommand{\Rule}[3]{
            \begin{minipage}[t]{#1}
                \begin{minipage}[t]{\linewidth}
                    #2
                \end{minipage}
                \hrule
                \begin{minipage}[t]{\linewidth}
            #3
                \end{minipage}
            \end{minipage}
}
\newcommand{\Ruledef}[5]{
                \begin{minipage}[c]{#3}
                    \hspace{1cm}{\bf #2}
                    \Rule{#3}{#4}{#5}
                \end{minipage}
}

\newcommand{\RulePL}{{\textit{\textbf{UUL}}}}
\newcommand{\RulePR}{{\textit{\textbf{UUR}}}}
\newcommand{\RuleNL}{{\textit{\textbf{DDL}}}}
\newcommand{\RuleNR}{{\textit{\textbf{DDR}}}}
\newcommand{\RulePN}{{\textit{\textbf{UD}} }}
\newcommand{\Type}[1]{{\mathit{Type}(#1)}}

\newenvironment{proofsketch}{\begin{proof}[sketch]}{\end{proof}}

\title{Trace Inclusion for One-Counter Nets Revisited}
\author{Piotr Hofman\inst{1}\and Patrick Totzke\inst{2}}
\institute{
University of Bayreuth, Germany
\and
LaBRI, Universit\'e de Bordeaux
}

\begin{document}
    \maketitle

    \begin{abstract}
      \TextOCNs\ (\OCN) consist of a nondeterministic finite control
and a single integer counter that cannot be fully tested for zero.
They form a natural subclass of both One-Counter Automata, which allow
zero-tests and Petri Nets/VASS, which allow multiple such weak
counters.
The \tincp\ has recently been shown to be undecidable
for \OCN.
In this paper, we contrast the complexity of two natural restrictions
which imply decidability.

First, we show that
\tinc\ between an \OCN\ and a \emph{deterministic} \OCN\
is NL-complete,
even with arbitrary binary-encoded initial counter-values as part of the input.
%
Secondly, we show
Ackermannian completeness of for the trace \emph{universality}
problem of nondeterministic \OCN.
This problem is equivalent to checking \tinc\ between a finite
and a \OCN-process.

    \end{abstract}
   
    \section{Introduction}
    A fundamental
question
in formal verification is
if the behaviour of one process can be reproduced by --
or equals that of -- another given process.
These inclusion and equivalence problems, respectively
have been studied for various
notions of behavioural preorders and equivalences
and for many
computational models.
\Tinc/equivalence asks if the set of \emph{traces}, all emittable
sequences of actions, of one process is contained in/equal to
that of another.
Other than for instance Simulation preorder, \tinc\
lacks a strong locality of failures, which
makes this problem intractable or even undecidable
already for very limited models of computation.

We consider \textOCNs, which
consist of a finite control and a single
integer counter that cannot be fully tested for zero,
in the sense that an empty counter can only restrict
possible moves.
They are subsumed by \TextOCAs\ (\OCA) and thus Pushdown Systems,
which allow explicit zero-tests by reading a bottom marker on the stack.
At the same time, \OCN\ are a subclass of Petri Nets
or Vector Addition Systems with states (VASS):
they are exactly the one-dimensional VASS and thus equivalent to
Petri Nets with at most one unbounded place.

%

\emph{Related work}.
\Textcite{VP1975} show the decidability of the trace equivalence problem for
\emph{deterministic} \textOCAs\ (\DOCA).
This problem has recently been shown to be NL-complete by
\textcite{BGJ2013}, assuming fixed initial counter-values.
%
The equivalence of \textDPDAs\ is known to
be decidable \cite{Sen1997} and
primitive recursive \cite{Sti2002},
but the exact complexity is still open.

\Textcite{Val1973} proves the undecidability of both \tinc\ for
\DOCA\ and universality for nondeterministic \OCA.
%
%
%
\Textcite{JEM1999} consider
\tinc\ between Petri Nets and finite systems and prove
decidability in both directions.
\Textcite{Jan1995}
showed that \tinc\ becomes undecidable if one compares
processes of Petri Nets with at least two unbounded places.
In \cite{HMT2013}, the authors show that
\tinc\ is undecidable already for (nondeterministic) \textOCNs.
Simulation preorder however,
is known to be decidable
and \PSPACE-complete for 
this model \cite{AC1998,JKM2000,HLMT2013},
which implies a \PSPACE\ upper bound for \tinc\
on \DOCN\ as \tinc\ and simulation coincide
for deterministic systems.

%
\Textcite{HWT1996} compare the classes of \emph{languages} 
defined by \DOCN\ with various acceptance modes
and in a series of papers consider the respective inclusion problems.
%
They derive procedures that exhaustively search for a bounded witness
that work in time and space polynomial in the size of the automata
if the initial counter-values are fixed.
We show that for monotone relations like \tinc\ or the inclusion
of languages defined by acceptance with final states,
one can speed up the search for suitable witnesses.

\emph{Our contribution}.
We fix the complexity of two well-known decidable 
decision problems regarding the traces of one-counter processes.

First, we show that \tinc\ between \emph{deterministic} \textOCN\
is NL-complete. Our upper bound holds even if only the supposedly
larger process is deterministic and if (binary encoded) initial
counter-values are part of the input.
This matches
the trivial NL lower bound derived from DFA universality.
Our technique uses short certificates for the existence of (possibly
long) distinguishing traces. The sizes of certificates are polynomial
in the number of states of the finite control and they can be verified
in space logarithmic in the binary representation of the initial
counter-values.

Our second result is that \tuniv\ of \emph{nondeterministic} \OCN\
is Ackermann-complete.
This problem can be easily seen to be (logspace) inter-reducible
with checking \tinc\ between a finite process and a process of a \OCN.
%

    \section{Background}
    We write $\N$ for the set of non-negative integers.
For any set $A$, let $A^*$ denote the set of finite strings over $A$
and $\varepsilon\in A^*$ the empty string.
\begin{definition}[One-Counter Nets]
    A \emph{\textOCN} (\OCN) is given as triple $\NN{N}=(Q,\Act,\delta)$
    where $Q$ is a finite set of
    control-states, $\Act$ is a finite set of action labels and
    $\delta\subseteq Q\x \Act\x\{-1,0,1\}\x Q$ is a set of transitions,
    each written as $p\step{a,d} p'$.
    A \emph{process} of $\NN{N}$ consists of a state $p\in Q$
    and a counter-value $m\in\N$. We will simply write $pm$
    for such a pair.
    Processes can evolve according to the transition rules of the net:
    For any $a\in\Act$, $p,q\in Q$ and $m,n\in\N$
    there is a step $pm\step{a}qn$ iff there exists
    $(p\step{a,d}q)\in\delta$ such that
    \begin{equation}
        n=m+d\ge0.
    \end{equation}

   The net $\NN{N}$ is \emph{deterministic} (a \DOCN) if for every $p\in Q$ and $a\in\Act$, there
    is at most one transition $(p,a,d,q)\in\delta$.
    It is \emph{complete} if for every $p\in Q$ and $a\in \Act$
    at least one transition $(p,a,d,q)\in\delta$ exists.
\end{definition}


In this paper we will w.l.o.g.~consider input nets in a certain normal form.
Specifically,
we assume what are sometimes called \emph{realtime}
automata, in which no silent ($\varepsilon$-labelled) transitions are present.
In the absence of zero-tests, the usual syntactic restriction
for deterministic pushdown automata, 
that no state with outgoing $\varepsilon$-transition
may have outgoing transitions labelled by $a\neq\varepsilon$
implies that all states on $\varepsilon$-cycles are essentially deadlocks
and one can eliminate $\varepsilon$-labelled transitions in logarithmic space.

\begin{definition}[Traces]
    Let $pm$ be a process of the \OCN\ $\NN{N}$.
  The \emph{traces} of $pm$ are the elements of the set
  \begin{equation*}
      T_\NN{N}(pm)=\{a_1a_2\dots a_k\in\Act^*\ |\
          \exists qn\  pm\step{a_0}\circ\step{a_1}\circ\dots\circ\step{a_k}qn\}.
  \end{equation*}
  We will omit the index $\NN{N}$ if is clear from the context.
  \emph{\Tinc} is the decision problem that asks if $T_\NN{A}(pm)\subseteq
  T_\NN{B}(p'm')$
  holds for given processes $pm$ and $p'm'$ of
  nets $\NN{A}$ and $\NN{B}$, respectively.
  \emph{\Tuniv} asks if $\Act^*\subseteq T(\alpha)$ holds for a given
  process $pm$.
\end{definition}

An important property of \textOCNs\ is that the step relation
and therefore also \tinc\ is monotone with respect to the counter:
\begin{lemma}[Monotonicity]\label{lem:monotonicity}
If $pm\step{a}p'm'$ then $p(m+1)\step{a}p'(m'+1)$.
This in particular means that $T(pm)\subseteq T(p(m+1))$
holds for any \OCN-process $pm$.
\end{lemma}

The next lemma justifies our focus on processes of complete \OCN. 
The proof is a simple construction and can be found in \cref{app:reduction}.
The idea is to first determinize $\NN{A}$ by consistently relabelling all transitions
of $\NN{A}$ and $\NN{B}$, and then complete the net $\NN{B}$ by introducing
a sink state.
\begin{restatable}[Normal Form Assumption]{lemma}{lemreduction}
\label{lem:reduction}
    \Tinc\ for \OCN\ is logspace-reducible to \tinc\ between a determinisic and
    a complete \OCN.
    More precisely, given \OCNs\ $\NN{A}$ and $\NN{B}$
    with state sets $N$ and $M$,
    one can construct a
    \DOCN\
    $\NN{A'}$ with states $N$
    and a complete \OCN\ $\NN{B'}$
    with states $M'\supseteq M$
    such
    that the following holds for
    any two processes $pm$ and $qn$ of $\NN{A}$ and $\NN{B}$,
    respectively:
    \begin{equation}
        T_\NN{A}(pm)\subseteq T_\NN{B}(qn) 
        \iff T_\NN{A'}(pm)\subseteq T_\NN{B'}(qn). 
    \end{equation}
    Moreover,
    the constructed net $\NN{B'}$ is deterministic if the original net $\NN{B}$ is.
\end{restatable}
Due to the undecidability of \tinc\ for \OCN\ \cite{HMT2013},
a direct consequence of \cref{lem:reduction} is that \tinc\
$T_\NN{A}(pm)\subseteq T_\NN{B}(qn)$
is already undecidable if we allow
the net $\NN{B}$ to be nondeterministic.
Unless otherwise stated, we will from now on assume a
\DOCN\ $\NN{A}=(Q_A,\Act,\delta_A)$ and
a complete \DOCN\ $\NN{B}=(Q_B,\Act,\delta_B)$.

    \section{Trace Inclusion for Deterministic One-Counter Nets}
    We characterize
witnesses for non-inclusion
$T_\NN{A}(pm)\not\subseteq T_\NN{B}(qn)$,
starting with
some notation to express
paths and their effects.

%
\begin{definition}[\OCN\ Paths] 
    \label{def:ocnpaths}
    Consider the OCN $\NN{N}=(Q,\Act, \delta)$.
    For the transition $t=(p,a,d,p')\in \delta$
    we write $source(t)=p$, $target(t)=p'$ and $\Delta(t)=d$ 
    for its source and target states and counter-effect, respectively.
    A \emph{path} in $\NN{N}$ 
    is a sequence $\pi=t_0t_1\dots t_k\in \delta^*$ of transitions where
    $target(t_i)=source(t_{i+1})$ for every $i<k$.
    Let $\prefix{i}{\pi}$ denote its prefix of length $i$.
    The \emph{effect} $\effect{}(\pi)$ and \emph{guard} $\guard{}(\pi)$
    of $\pi$ are
    \begin{equation*}
        \effect{}(\pi) = \sum_{i=0}^k \Delta(t_i)
        \qquad\text{and}\qquad
        \guard{}(\pi) = - \min\{\effect{}(\prefix{i}{\pi})\ |\ 0\le i\le k\}.
    \end{equation*}
    The path $\pi$ is \emph{enabled} in process $pm$ (write $pm\step{\pi}$)
    if $\guard{}(\pi)\le m$. The source and target nodes of $\pi$ are those
    of its first and last transition, respectively.
    We write $pm\step{\pi}p'm'$ if $\pi$ takes $pm$ to $p'm'$,
    i.e., if $pm\step{\pi}$, $target(\pi)=p'$ and $m'=m+\effect{}(\pi)$.
\end{definition}

The guard $\guard{}(\pi)$ is the minimal counter-value
that is sufficient to
traverse the path $\pi$ while
maintaining a non-negative counter-value along the way.
This value is always non-negative.

Notice that
the absolute values of the effect and guard of a path
are bounded by its length.
We consider the synchronous product of the control graphs of
two given \textDOCNs.
\begin{definition}[Product Paths]
    The \emph{product}
    of nets $\NN{A}$ 
    and $\NN{B}$ 
    is the finite
    graph
    with nodes 
    $V=Q_\SP\x Q_\DUP$ and $(\Act\x\{-1,0,1\}\x\{-1,0,1\})$-labelled edges
    $E$, where 
    \begin{equation*}
      (p,q)\xrightarrow{a,d_\SP,d_\DUP}(p',q')\in E
      \text{ iff } p\xrightarrow{a,d_\SP}p' \in\delta_\SP
      \text{ and } q\xrightarrow{a,d_\DUP}q'\in \delta_\DUP.
    \end{equation*}

    A \emph{path} in the product is a sequence
    $\pi=T_0T_1\dots T_k\in E^*$
    and defines paths $\pi_\SP$ 
    and $\pi_\DUP$ 
    in nets $\NN{A}$ and $\NN{B},$ respectively.
    It is \emph{enabled} in $(pm,qn)$ 
    if $\pi_\SP$ and $\pi_\DUP$
    are enabled in $pm$ and $qn,$ respectively. 
    In this case we write
    $(pm,qn)\step{\pi}(p'm',q'n')$ to mean that
    $pm\step{\pi_\SP}p'm'$
    and $qn\step{\pi_\DUP}q'n'$.
    %
    We lift the definitions of source and target nodes 
    to paths in the product:
    $source(\pi)=(source(\pi_A),source(\pi_B))\in V$,
    $target(\pi)=(target(\pi_A),target(\pi_B))\in V$.
    Moreover, write
    $\effect{A}(\pi)$,
    $\effect{B}(\pi)$,
    $\guard{A}(\pi)$ and 
    $\guard{B}(\pi)$
    for the effects and guards of $\pi$ in nets $\NN{A}$ and $\NN{B},$
    respectively.
\end{definition}


Since both $\NN{A}$ and $\NN{B}$ are deterministic and
$\NN{B}$ is complete,
a trace $w\in T_\NN{A}(pm)$ %
uniquely determines a path from state $(p,q)$ in their product.
We therefore identify witnesses for non-inclusion
with the 
paths they induce in the product.
\begin{definition}[Witnesses]
    Assume $T_{\NN{A}}(pm)\not\subseteq T_{\NN{B}}(qn)$ for processes $pm$ and $qn$ of
  $\NN{A}$ and $\NN{B}$.
  A \emph{witness} for $(pm,qn)$ is a path $\pi$ in the product
  of $\NN{A}$ and $\NN{B}$
  such that
  $(pm,qn)\Step{\pi}{}{}(p'm',q'n')$
  and for some $a\in\Act$,
  $p'm'\step{a}$ but $q'n'\notStep{a}{}{}$.
\end{definition}

Every witness $\pi$ for $(pm,qn)$ 
completely exhausts the counter in the process of $\NN{B}$:
$(pm,qn)\Step{\pi}{}{}(p'm',q'0)$.
This is because a process of a complete net can only \emph{not} make an $a$-step
in case the counter is empty.

\begin{example}
    Consider two nets 
    given by 
    self-loops
    $p\step{a,0}p$
    and $q\step{a,-1}q,$ respectively.
    Their product is the cycle $L=(p,q)\xrightarrow{a,0,-1}(p,q)$
    with effects $\effect{A}(L)=0$ and $\effect{B}(L)=-1$.
    The only witness for $(pm,qn)$ for initial counter-values $m,n\in\N$ is
    $L^{n}$, which has length 
    polynomial in the sizes of the nets
    \emph{and} the initial counter-values, but not in the sizes
    of the nets alone.
\end{example}
The previous example shows that if binary-encoded initial counter-values
are part of the input, we can only bound the length of shortest witnesses
exponentially. 
However, we will see that it suffices to consider witnesses
of a certain regular form only.
This leads to
small certificates for non-inclusion,
which can be stepwise guessed and verified in
space
logarithmic
in the size of the nets.

A crucial ingredient for our characterization is the 
monotonicity of witnesses, a direct consequence of
the monotonicity
of the steps in \OCNs\ (\cref{lem:monotonicity}):
\begin{lemma}\label{lem:witness-monotonicity}
    If $\pi$ is a witness for $(pm,qn)$
    then for all $m'\ge m$ and $n'\le n$ some prefix of $\pi$
    is a witness for $(pm',qn')$.
\end{lemma}
%

The intuition behind the further characterization of witnesses
is that in order to show non-inclusion, one looks for a path
that is enabled
in the process of $\NN{A}$ and moreover exhausts the counter
in the process of $\NN{B}$.
Since any sufficiently long
path will revisit control-states in the product,
we can compare such paths with respect
to their effect on the counters and see that some are ``better'' than
others. For instance, a cycle that only increments the counter in  $\NN{B}$
and decrements the one in $\NN{A}$ is surely suboptimal considering our
goal to find a (shortest) witness.
The characterization \cref{thm:form} essentially states that
if a witness exists, then also one that, apart from short paths,
combines only the most productive cycles.

\newcommand{\TypeDD}{(<,<)}
\newcommand{\TypeUU}{(>,\ge)}
\newcommand{\TypeDU}{(\le,\ge)}
\newcommand{\TypeUD}{(\ge,<)}

\begin{definition}[Loops]
%
    A non-empty path $\pi$
    in the product
    is called a \emph{cycle} if $source(\pi)=target(\pi)$.
    Such a cycle is a \emph{loop} if none of its proper subpaths is a cycle.
    The \emph{slope} of loop $\pi$ is the ratio $S(\pi)=\effect{\SP}(\pi)/\effect{\DUP}(\pi)$,
    where for $n>0$ and $k\in\Z$ we let $n/0=\infty>k$, $0/0=0$ and $-n/0=-\infty<k$.
    Based on the effect of a loop we distinguish four types of loops: $\TypeDD,\TypeUU,\TypeDU$,
    and $\TypeUD$. The \emph{type} of $\pi$ is $Type(\pi)=(\blacktriangleleft,\blacktriangleright)$ iff
    $\effect{\SP}(\pi)\blacktriangleleft 0$ and $\effect{\DUP}(\pi)\blacktriangleright 0$.
\end{definition}

Note that no loop is longer than $|V|$
because it visits exactly one node twice.
\begin{example}
    \label{ex:witness}
    Consider two \DOCN\ such that their product is the graph depicted
    below, 
    where we identify transitions with their action labels for simplicity and

\noindent
\begin{minipage}{\linewidth}
    \begin{wrapfigure}[9]{r}{0.430\linewidth}
        \vspace{-0.9cm}
  \includegraphics[width=0.97\linewidth]{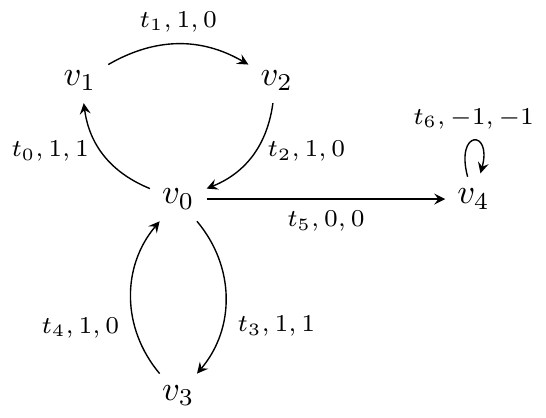}
    \end{wrapfigure}
    let $v_0=(p,p')\in V$.
    The paths $t_0t_1t_2$,
    $t_3t_4$ and $t_6$ are loops with
    slopes $3/1$, $2/1$ and $1/1$
    and types $\TypeUU$, $\TypeUU$ and $\TypeDD,$
    respectively.
  The path $(t_0t_1t_2)(t_3t_4)^9t_5(t_6)^{20}$ is a witness for $(p0,p'10)$
    of length $42$.
    By replacing $8$ occurrences of the loop $(t_3t_4)$ with
    $(t_0t_1t_2)^8$ we derive the longer witness
    $(t_0t_1t_2)^9(t_3t_4)t_5(t_6)^{20}$, which has
    essentially the same structure but is more efficient in the sense
    that for the same effect
    on $\NN{B}$ it achieves a higher counter-effect on $\NN{A}$.
%

%
\end{minipage}
\end{example}
%
\begin{restatable}[]{theorem}{thmform}
\label{thm:form}
    Fix a \DOCN\ $\NN{A}$, 
    a complete \DOCN\ $\NN{B}$, 
    and
    let $K\in\N$ be the number of nodes in their product.
    There is a bound $c\in\N$ that depends polynomially on $K$,
    such that
    the following holds
    for any two processes $pm$ and $qn$ of $\NN{A}$ and $\NN{B}$.
    If $T(pm)\not\subseteq T(qn)$,
    then there is a witness for $(pm,qn)$
    that is either no longer than $c$ or has
    one of the following forms:
    \begin{enumerate}
    \item $\pi_0L_0^{l_0}\pi_1$, where
            $L_0$ is a loop of type $\TypeUD$ and $\pi_0,\pi_1$ are no longer than $c$,
        \item $\pi_0L_0^{l_0}\pi_1L_1^{l_1}\pi_2$, where
            $L_0$ and $L_1$ are loops of type $\TypeUU$ and $\TypeDD$ with $S(L_0) > S(L_1)$
            and $\pi_0,\pi_1,\pi_2$ are no longer than $c$,
        \item $\pi_0L_0^{l_0}\pi_1$, where
            $L_0$ is a loop of type $\TypeDD$
            and $\pi_0,\pi_1$ are no longer than $c$,
    \end{enumerate}
    where in all cases, the number of iterations $l_0,l_1\in\N$
    are polynomial in $K$ and the initial
    counter-values $m$ and $n$ of the given processes.
\end{restatable}
\begin{proofsketch}
    The overall idea of the proof is to explicitly rewrite
    witnesses into one of the canonical forms.
    More specifically,
    we introduce a system of path-rewriting rules
    which simplify witnesses by removing, reducing or changing some
    loops as in \cref{ex:witness}.
    We show that the rules preserve witnesses and
    any sequence of successive rule applications must eventually terminate
    with a normalized path, to which none of the rules is applicable.
    Such a witness can be decomposed as
    \begin{equation}
        \label{eq:form:sketch}
        \pi=\pi_0L_0^{l_0}\pi_{1}L_{1}^{l_1}\dots\pi_kL_k^{l_k}\pi_{k+1}
    \end{equation}
    where the $L_i$ are (pairwise different) loops and the 
    $\pi_i$ are short, i.e.~polynomially bounded.
 Moreover the rules are designed in such a way that almost all $l_i$ are polynomially bounded. By almost all we mean except one in the first and third form of the witness or two in the witness of the second form. This means that unravelling of those loops with polynomially bounded $l_i$ and glueing them with
surrounding $\pi_i$ to get paths $\pi_0,\pi_1,\pi_2$ does not blow up of the length of $\pi_0,\pi_1, \pi_2$ above polynomial
 bound $c$.
\qed
\end{proofsketch}

Notice that the bound $c$ in the claim of \cref{thm:form}
depends only on the number of states.
We now derive a decision procedure
for \tinc\ that works in logarithmic space.
%


\begin{theorem}
    Let $pm$ and $qn$ be processes of \OCN\ $\NN{A}$
    and \DOCN\ $\NN{B},$ 
    respectively, where $m,n$ are given in binary.
    There is a nondeterministic algorithm that
    decides $T(pm)\subseteq T(qn)$ in
    logarithmic space.
\end{theorem}
\begin{proof}
    Let 
    $\NN{A}=(Q_\SP,\Act,\delta_\SP)$ and
    $\NN{B}=(Q_\DUP,\Act,\delta_\DUP)$,
    and let $K\in\N$ be the number of states in their product.
    By \cref{lem:reduction}, we can assume w.l.o.g.\ that $\NN{A}$ is 
    deterministic and $\NN{B}$ is complete and deterministic
    and 
    so \cref{thm:form} applies.

If the initial counter-values are $m=n=0$, \cref{thm:form} implies a polynomial bound on the length of shortest witnesses.
In that case, one can simply stepwise guess and verify a witness,
explicitly storing the
intermediate processes with binary encoded counter-values
in logarithmic space.
Such a procedure is impossible with arbitrary initial counter-values
as part of the input, because one
does not even have the space to memorize them.

    For the general case, we argue that one can nondeterministically guess
    a template (consisting of short paths) and verify in logspace that
    there is indeed some witness that fits this template.
    \Cref{thm:form} allows us to either guess a short ($\le c\in poly(K)$) witness
    or one of forms 1,2 or 3, together with matching short paths $\pi_i,L_i$.
    The effect and guard of these paths
    are bounded by their lengths and hence by $c$.
    This means $\O(\log K)$ space suffices to stepwise compute
    the binary representation of these values
    and verify that the conditions the form
    imposes on the types and slopes of the loops are met.
%
    It remains to check if 
    exponents $l_0,l_1\in\N$ exist, that complete the description of a
    witness $\pi$.
    To see why these checks can be implemented in logarithmic space,
    first recall
    that one can verify inequalities of the form
\begin{equation}
  \label{eq:app_inequality}
  m\cdot {A} +{B} \ge n \cdot {C} +{D}
\end{equation}
in $\O(log(A+B+C+D))$ space, if $m,n\in \N$ are given in binary
(see \cref{app:inequalities}).

    For templates of the first two forms, it suffices to 
%
    check if $m\ge\guard{\SP}(\pi_0L_0)$, because the type of $L_0$ implies that
    $\guard{\SP}(\pi_0L_0^l)\le \guard{\SP}(\pi_0L_0)$ for all $1<l\in\N$.
    This means that the process $pm$ of $\NN{A}$ can go to, and repeat the loop $L_0$
    arbitrarily often. In case its effect in $\NN{B}$ is negative
    (in templates of form 1), this immediately
    implies the existence of a suitable $l_0$. For templates of form 2)
    the existence of $l_0,l_1\in\N$ completing the description of a witness
    is guaranteed because the slope of the first loop is bigger than that of the second.

    For templates of the third kind
    recall that, because $\NN{B}$ is complete, a path $\pi=\pi_0L_0^{l_0}\pi_1$
    is a witness iff there is some edge $T$
    in the product such that $\effect{\DUP}(T)=-1$ and
    both $m\ge\guard{\SP}(\pi T)$ and $n+\effect{\DUP}(\pi T) = -1$.
    Equivalently, we can write this as
%
    \begin{align}
        m+\effect{\SP}(\pi_0L_0^{l_0})&=m+\effect{\SP}(\pi_0)+\effect{\SP}(L_0)\cdot l_0 \ge
        \guard{\SP}(\pi_1 T)\text{ and }\label{eq:main:m}\\
        n+1 &= - \effect{\DUP}(\pi T) = -\effect{\DUP}(\pi_0) - \effect{\DUP}(L_0) \cdot l_0 - \effect{\DUP}(\pi_1 T)
        \label{eq:main:n}.
    \end{align}
    Eliminating $l_0$, we see that this is true iff
    \begin{equation}
        m+\effect{\SP}(\pi_0) + \effect{\SP}(L_0)\cdot
            \frac{
            \effect{\DUP}(\pi_0)+\effect{\DUP}(\pi_1)+n
        }{
            - \effect{\DUP}(L_0)
        }
       \ge \guard{\SP}(\pi_1).
    \end{equation}
    Simplifying further we can bring this into the form $m\cdot A - n\cdot B \ge C$
    where $A,B,C$ are polynomial in $c$.
    The condition can be checked in $\O(\log K)$ space.
    \qed
\end{proof}

%
%

    \section{Universality of Nondeterministic One-Counter Nets}
    \label{sec:universality}
    To contrast the result of the previous section we now turn
to the problem of checking \tinc\ between a finite process
and a nondeterministic \OCN.
This problem is known to be decidable, even for
general Petri nets \cite{JEM1999}
and it 
can be easily seen
to be (logspace) inter-reducible with
the trace universality problem, because
\OCNs\ are closed under products with finite systems.

For \OCN,
trace universality
can be decided using a
simple well-quasi-order based saturation method
that determinizes the net on the fly.
We will see that this procedure is optimal:
The problem is Ackermannian, i.e.~it
is non-primitive recursive and lies exactly
at level $\omega$ of the Fast Growing Hierarchy \cite{FFSS2011}.

%
\newcommand{\Nb}{\N_{\bottom}}
Let $\Nb$ be the set of non-negative integers plus
a special least element $\bottom$ and let $\max$
be the total function that returns the maximal element
of any nonempty finite subset and $\bottom$ otherwise.
Consider a set $S\subseteq Q\x\N$ of processes of an \OCN\ $\NN{N}=(Q,\Act,\delta)$.
We lift the definition of traces to sets of processes in the natural way:
the \emph{traces} of $S$ are $T(S)=\bigcup_{qn\in S} T(qn)$.
By the monotonicity of \tinc\ (\cref{lem:monotonicity}),
the traces of a finite set of processes are determined
only by the traces of its maximal elements.

\begin{definition}
    Let $Q=\{q_1,q_2,\dots,q_k\}$ be the states-set of some \OCN.
    For a finite set $S\subseteq Q\x\N$
    define the \emph{macrostate} as the vector
    $M_S 
    \in \Nb^k$ 
    where for each $0< i\le k$,
    $M_S(i)= M_S(q_i) = \max\{n\;|\; q_in \in S\}$. 
    In particular, the macrostate for a singleton set $S=\{q_in\}$
    is the vector with value $n$ at the $i$-th coordinate
    and $\bot$ on all others.
    The \emph{norm} of a macrostate $M \in \Nb^k$
    is $\norminf{M} = \max\{M(i) \;|\; 0<i\leq k\}.$
    We define a step relation $\WStep{a}{}{}$
    for all $a\in\Act$ on the set of macrostates as follows:
    \begin{equation}
    (n_1,n_2,\dots,n_k)\WStep{a}{}{} (m_1,m_2,\dots,m_k)
    \end{equation}
    iff
    $m_i = \max\{n\:|\: \exists n_j\neq\bottom.\: q_jn_j\step{a}q_in \}$
    for all $0< i\le k$.
%
    The \emph{traces} of macrostate $M$ are
    $T(M) = \bigcup_{0< i\le k} T(q_{i\, } M(i))$,
    where $T(q\bot)=\emptyset$.
    For two macrostates $M,N$ we say $M$ is \emph{covered} by $N$ and write
    $M\sqsubseteq N$,
    if it is pointwise smaller, i.e., $M(i)\le N(i)$ for all
    $0< i\le k$.
    For convenience, 
    we will write $\{q_1=n_1, q_2=n_2,\dots,q_l=n_l\}$
    for the macrostate with value $M(i)=n_i$
    whenever $q_i=n_i$ is listed
    and $\bot$ otherwise.
\end{definition}

Steps on macrostates correspond
to the classical powerset construction and each macrostate
represents the finite 
set of possible processes the \OCN\ can be in,
where all non-maximal ones (w.r.t.~their counter-value) are pruned out.

\begin{example}{Macrostate}

\begin{minipage}[b]{0.35\textwidth}

 \begin{tikzpicture}[node distance=2cm]
            \node (P) at (0,0) {$q2$};
            \node (Q) at (0, 3) {$q1$};
            \node (R) at (2,1.5) {$q3$};   
            \node (nic) at (0,-2) {};
              
            \path[->]
                  (P) edge node[left] {$a, 1$} (Q)
                  (Q) edge node[right] {$a, 0$} (R)
                  (R) edge node[right] {$a, -1$} (P);
     \path[->]
		(R) edge [loop right] node[right] {$a,1$} (R);           
       \end{tikzpicture}

\end{minipage}
\begin{minipage}[b]{0.6\textwidth}
Consider automaton $\NN{A}$ like on the picture, state $q_3$ and a counter value $4$; we analyse traces, $T(q_3 4)$. If we go via an edges labelled by $a$ once we can see that $T(q_3 4)=\{\varepsilon\}\, \cup\, aT(q_2 3)\, \cup\, aT(q_35)$. This implies that $T(q_3 4)$ is universal iff $T(q_3, 5)\cup T(q_2, 3)$ is universal, i.e. contains $\Act^*$. Making similar analysis after using two more $a$ we get that $T(q_3 4)$ is universal iff $T(q_37)\, \cup \, T(q_2 5)\, \cup \, T(q_1 5)\, \cup \, T(q_3 4)$ is universal. But we know that $T(q_3 4)\subseteq T(q_3 7)$ which implies that
$T(q_37)\, \cup \, T(q_2 5)\, \cup \, T(q_1 5)\, \cup \, T(q_3 4)=T(q_37)\, \cup \, T(q_2 5)\, \cup \, T(q_1 5)$. This immediately lead to introduce macrostates $M_{\{q_3 7,q_2 5,q_1 5,q_ 34\} }= (5,5,7).$ 
The norm $\norminf{M_{\{q_3 7,q_2 5,q_1 5,q_ 34\} }}= 7.$
On the other hand $M_{\{q_3 4\} }=(\bottom, \bottom, 4)$ which means that 
states $q_1$ and $q_2$ are not present and in this case $M(1)=M(q_1)=\bottom$. Moreover we can write that
$M_{\{q_ 34\} }\sqsubseteq M_{\{q_3 7,q_2 5,q_1 5,q_ 34\} }.$

\end{minipage}

\end{example}

The next lemma directly follows
from these definitions and monotonicity (\cref{lem:monotonicity}).
\begin{lemma}\label{lem:macroprops}\
\begin{enumerate}
  \item The covering-order $\sqsubseteq$ is a well-quasi-order on $\Nb^k$, the set
      of all macrostates. Moreover, $M\sqsubseteq N$ implies $T(M)\subseteq T(N)$.
      \label{lem:macroprops:wqo}
  \item If $M\wstep{a}N$ then $\norminf{N} \le \norminf{M}+1$.
      \label{lem:macroprops:norminc}
  \item For any finite set $S\subseteq Q\x\N$ it holds that $T(S) = T(M_S)$.
      \label{lem:macroprops:traces}
\end{enumerate}
\end{lemma}

Dealing with macrostates allows us to
treat universality as a reachability problem:
By point~\ref{lem:macroprops:traces}
of \cref{lem:macroprops} we see that a
process $qn$ is \emph{not} trace universal,
$\Act^* \neq T(qn)$, if and only if $M_{\{qn\}} \WStep{}{*}{} (\bottom,\bot,\dots,\bot)$.
We take the perspective of a pathfinder, whose goal it is to
reach $(\bot)^k$. 

We can decide universality by
stepwise guessing a shortest terminating path from the initial macrostate,
and thus a witness for non-universality.
Whenever we see a macrostate that covers one of its predecessors,
we can safely discard this candidate,
because
omitting the intermediate path would result in a shorter witness
by \Cref{lem:macroprops}.\ref{lem:macroprops:wqo}.

\newcommand{\inc}{\mathrm{inc}}
\newcommand{\dec}{\mathrm{dec}}
\newcommand{\ifz}{\mathrm{ifz}}
We show non-primitive recursiveness by reduction from the
control-state reachability problem for incrementing \textCMs\
\cite{DL2009,FFSS2011}.
\begin{definition}[\TextCMs]
    A (Minsky)-\textCM\ (\CM) 
    is an automaton with finitely many states $Q$,
    finitely many \emph{counters} $C_1,C_2,\dots,C_k$,
    and transitions are of the form $Q \x \Act \x Q$ where $\Act$ is
    $\{\inc, \dec, \ifz\} \x\{ 1,2,\dots,k\}$.
    A \emph{configuration} of such a \CM\
    consists of
    a state and a valuation of the counters.
    Performing a transition $(p, (op,i), q)$ changes a configuration precisely: the state changes from $p$ to $q$ and we make operation $op$ on the counter $c_i$,
    where $\inc, \dec$ and $\ifz$ mean increment, decrement and zero-test,
    respectively.
    Such a step is forbidden if the requested operation is $\dec$
    and the value of $c_i$ is $0$, or if $c_i>0$ and the operation is $\ifz$.

    An \emph{\textIncCM} (\incCM) is a \CM\
    in which counters can spontaneously
    increment without performing any transitions.
    Such increments we call incrementing errors.
    \emph{Control-state reachability} is the decision problem that
    asks if there is a run of a given
    \CM\ from an initial configuration to some given state $q_f\in Q$.
\end{definition}


%
%
%
Our reduction is based on the following simple observation.
Consider an \OCN\ $\NN{N}=(Q,\Act,\delta)$ that contains a \emph{universal} state
$U$: it has self-loops $U\step{a,0}U \in\delta$ for every action $a\in\Act$.
A Pathfinder who wants to prove non-universality
must avoid macrostates with $M(U)\neq \bot$, because
no continuation of a path leading to such a macrostate can be a witness.
We can use this idea to construct macrostates that
prevent Pathfinder from making certain actions.

\begin{definition}[Obstacles]
  Let $S\subseteq\Act$ be a set of actions in an \OCN\
  that contains a universal state $U$.
  A state $q\in Q$ is called an \emph{S-obstacle} if $q\step{a,0}U\in \delta$ for all
  actions $a\in S$.
  We say $q$ \emph{ignores} $S$, if $q\step{a,0}q\in \delta$ for all $a\in S$.
\end{definition}
Note that if a macrostate contains an $S$-obstacle, then Pathfinder
must avoid all actions of $S$. 
In order to remove an obstacle, Pathfinder must play an action that is not
the label of any of its incoming transitions.
\newcommand{\INIT}{\mathit{Init}}
\newcommand{\init}{\mathit{init}}
\newcommand{\Zero}{\mathit{Z}}
\newcommand{\LIM}{\mathit{Lim}}

\begin{theorem}
    \label{thm:univ-ack}
    \Tuniv\ for \OCN\ is not primitive recursive. 
\end{theorem}

\begin{proof}
By reduction from the control-state reachability problem for \incCM,
which has non-primitive recursive complexity \cite{DL2009}.
    We construct an \OCN-process $\INIT(0)$ that is not universal
    iff a given \incCM\ reaches a final state from its initial configuration.
    The idea is to enforce a faithful simulation of the \incCM\ by pathfinder,
    who wants to show non-universality of the \OCN\
    by stepwise rewriting the initial macrostate
    $\{\INIT=0\}$ to the all-bottom-macrostate $\bot^l$.
    
    We construct an \OCN\ $\NN{N}$ which has a unique action
    for every transition of the \incCM, as well as actions
    $\tau_i$ that indicate incrementing errors for every counter $c_i$,
    and actions $\sharp$ and $\$$ to mark the beginning and end of a run, respectively.
    This way we make sure there is a strict correspondence between words and
    \incCM-runs. The states of \NN{N} are
    \begin{itemize}
      \item a new initial state $\INIT$ and a universal state $U$,
      \item a state $q_i$ for every state $q_i$ of the \incCM,
      \item a state $C_i$ for every counter $c_i$ of the \incCM,
      \item a state $\Zero$, which ignores every action but the end marker $\$$. State $\Zero$ will be used to access the constant $0.$
    \end{itemize}

    A configuration $q(c_1,c_2,\dots,c_k)$ of the \incCM\ is represented by a macrostate
    $\{q=0, \Zero=0, C_1=c_1,C_2=c_2,\dots,C_k=c_k\}$.
    We will define the transitions of $\NN{N}$ such that
    the only way for Pathfinder
    to reach $\bot^l$
    is by
    rewriting the initial macrostate $\{\INIT=0\}$
    to the one representing the initial \incCM\ configuration
    and then to stepwise
    announce the transitions of an accepting run of the \incCM.
    %
Using the idea of obstacles, we define the rules of the net $\NN{N}$
so that the only way Pathfinder can avoid the universal state $U$
and reach the macrostate $\bot^k$ is by
first transforming the initial macrostate $\{\INIT=0\}$
to the one that represents the initial \incCM\ configuration
and then announcing
transitions (as well as actions demanding increment errors)
of a valid and accepting run of the \incCM.
%
    \paragraph{Initialization.}
    %
    %
    To set up $M_0=\{q_0=0,\Zero=0,C_0=0,C_1=0,\dots,C_k=0\}$, representing
    the initial \incCM\ configuration, we add 
    $\sharp$-labelled
    transitions
    with effect $0$ from $\INIT$ to $q_0,\Zero$ and $C_i$
    for all $0\le i\le k$.
    Moreover, we make $\INIT$ an obstacle for every action but $\sharp$.
    This way, Pathfinder has to play $\sharp$ as the first move
    (and set up $M_0$) in order to avoid a universal macrostate. 
    Furthermore we make $\#$ an obstacle for every state except of $\INIT$; this prevent playing $\#$ after the first move.

    \paragraph{Finite control.}
    For any transition $t=q\step{(a,i)}q'$ of the \incCM, we add a transition
    $q\step{t,0}q'$ to $\NN{N}$ that, in a macrostate-step, will replace
    the value $0$ in dimension $q$ by $\bot$ and introduce value $0$ in
    dimension $q'$.
    Moreover, we make every state $q$ an obstacle for all actions announcing
    \incCM-transitions not originating in $q$.
    This prevents Pathfinder from announcing transitions from $q$ unless
    the current macrostate has $M(q)=0$ and $M(q_i)=\bot$ for all $q_i\neq q$.


    \paragraph{Simulation of the Counters.}
    Every transition operates on one of the counters $c_i$ for $0\le i\le k$.
    Below we list the corresponding transitions 
    in the \OCN\ $\NN{N}$ for this counter.
    Every state of $\NN{N}$ not explicitly mentioned
    ignores the action in question. 
    In the macrostate, the values of these states are therefore unchanged.

\begin{description}
  \item[increments] For \incCM-transitions $t$ that increase the $i$th counter,
      $\NN{N}$ contains a $t$-labelled transition from state $C_i$ to $C_i$ with effect $+1$.
      Additionally, to deal with spontaneous increment errors,
      there is a $\tau_i$-labelled increasing self-loop in state $C_i$.
  \item[decrements] For \incCM-transitions $t$ that decrease the $i$th
      counter, $\NN{N}$ contains a $t$-labelled transition from state $C_i$ to $C_i$
      with effect $-1$.
      
      This means that the next macrostate $M$ could lose
      the value for this counter and have $M(C_i)=\bot$ if previously,
      the value was $0$.
      In that case, the decrementing step from value $0$ to value $0$ is
      valid in the \incCM\ because it can first (silently) increment and then
      do the (visible) decrement step.
      In order to avoid losing the state $C_i$ in the macrostate, the \OCN\
      contains a transition $\Zero\step{t,0}C_i$ from the
      constant-zero state $\Zero$ to state $C_i$.
      Recall that $\Zero$ is present in the macrostate because
      $\Zero$ ignores every action except the end marker $\$$.
      Consequently, no correctly set up macrostate will set $M(C_i)=\bot$.
  
  \item[zero-tests] For \incCM-transitions $t$ that test the $i$th counter
      for $0$, we add a $t$-labelled transition $C_i\step{t,-1}U$
      from state $C_i$ to the universal state.
      This prevents Pathfinder from using these actions if the current
      macrostate has $M(C_i)>0$ because it would make the next macrostate
      universal. If however $M(C_i)=0$, such a step is safe because
      the punishing transition is not enabled in the \OCN-process $C_i0$.
\end{description}
    
Lastly, we only add transitions to $\NN{N}$ so that the final state $q_f$
is the only original \incCM-state which is not an obstacle for $\$$.
This prevents Pathfinder from playing the end-marker $\$$ unless
the simulation has reached the final state.
\qed
\end{proof}

\begin{example}{Reduction.}

 \begin{minipage}[b]{0.35\textwidth}

  \begin{tikzpicture}[node distance=2cm]
           
                 \node (P) at (0,0) {$q0$};
            \node (Q) at (3, 1) {$q1$};
            \node (R) at (0,2) {$q2$};

            \node (nic) at (0,-0.5) {};
            \node (nic1) at (0,3.5) {};
               
            \path[->]
                  (P) edge node[below, yshift=-0.15cm] {$inc\ 1$} (Q)
                  (Q) edge node[above, yshift=0.15cm] {$dec\ 2$} (R)
                  (R) edge node[right] {$ifz\ 2$} (P);
           
       \end{tikzpicture}

\end{minipage}
\begin{minipage}[b]{0.6\textwidth}
Consider an incrementing error two counter machine (as on the left) and we ask about reachability from $q_0$ to $q_2.$

The one counter net which is result of the construction for the above reachability problem. We will use several universal states $U$ to avoid crossing arrows, moreover edges labelled with a sequence of labels mean a bunch of edges one for each label. We put labels into brackets, to clearly separate each label.
       
\end{minipage}

       \begin{tikzpicture}[node distance=3cm]
       \node (in) at (0,8) {$Init$};
       \node (Z) at (0,3) {$Z$};
       \node (C1) at (5,5) {$c_1$};
       \node (C2) at (5,3) {$c_2$};
       \node (U1) at (5,1) {$U$};
       \node (U2) at (0,10) {$U$};
       \node (U3) at (9,7) {$U$}; 
       
            \node (P) at (5,8) {$q0$};
            \node (Q) at (9,9 ) {$q1$};
            \node (R) at (5, 10) {$q2$};  
            \node (nic) at (5, 11.5) {};

            \path[->]
                  (P) edge node[below, yshift=-0.15cm] {$(a, 0)$} (Q)
                  (Q) edge node[above, yshift=0.15cm] {$(b, 0)$} (R)
                  (R) edge node[right] {$(c, 0)$} (P);
       
 \path[->]
                  (in) edge node[right, yshift=0.15cm] {$(\#, 0)$} (C1)
                  (in) edge node[right, yshift=0.15cm] {$(\#, 0)$} (C2)
                  (in) edge node[right] {$(\#, 0)$} (Z)		
		(in) edge node[below] {$(\#, 0)$} (P)
		(Z) edge node[below] {$(b, 0)$} (C2)
		(P) edge node[above] {$(\$, 0),\  (c,0),\ (b,0),\ (\#,0)$} (U2)
		(R) edge node[above] {$(\$, 0),\  (a,0),\ (b,0),\ (\#,0)$} (U2)
		(Q) edge node[right] {$(a,0),\ (c,0),\ (\#,0)$} (U3)
		(in) edge node[left] {$(\Act\setminus\{\#\}, 0)$} (U2)
		(C2) edge node[right] {$(c, -1),\ (\#,0)$} (U1)
		(C1) edge node[right, yshift=-0.3cm] {$(\#,0)$} (U3)
		(Z) edge node[below, yshift=-0.2cm] {$(\#,0)$} (U1);
		
		 \path[->]
		(C1) edge [loop right] node[right] {$(a,1),\ (\tau_1, 1),\ (\tau_2,0),\ (b,0),\ (c,0)$} (C1);      
		
		 \path[->]
		(C2) edge [loop right] node[right] {$(a,0),\ (\tau_1, 0),\ (\tau_2,1),\ (b,-1)$} (C1);    
		
		 \path[->]
		(Z) edge [loop left] node[below, yshift=-0.3cm] {$(\Act \setminus \{\$\}, 0)$} (Z);      
		  
		   \path[->]
		(P) edge [loop below] node[below] {$(\tau_1,0),\ (\tau_2, 0)$} (P);    
		   \path[->]
		(R) edge [loop above] node[above, xshift=1cm] {$(\tau_1,0),\ (\tau_2, 0)$} (R);    
		   \path[->]
		(Q) edge [loop right] node[right] {$(\tau_1,0),\ (\tau_2, 0)$} (Q);

       \end{tikzpicture}

\end{example}

For the rest of this section, we recall 
a recent result from \textcite{FFSS2011},
that allows us to
provide the exact complexity of the \OCN\ trace universality problem
in terms of its level in the Fast-Growing Hierarchy.
\begin{definition}[Fast-Growing Hierarchy] \label{def:fastgrowing-hierarchy}
    Consider the family of functions $F_n:\N\to\N$
    where for $x,k\in\N$,
    \begin{align*}
        F_0(x) = x+1 \text{\quad and }
        &&
        F_{k+1}(x) = F_k^{x+1}(x).
    \end{align*}
    Here, $F^k$ denotes the $k$-fold application of $F$.
    Moreover, define $F_{\omega}(x)=F_x(x)$ for the first limit ordinal $\omega$.
    For $k\le \omega$,
    $\mathfrak{F}_k$ denotes
    the least class of functions
    that contains all constants and is closed under
    substitution, sum, projections, limited recursion and applications
    of functions $F_n$ for $n\le k$.

    Already $\mathfrak{F}_2$ contains all elementary functions
    and the union $\bigcup_{k\in\N}\mathfrak{F}_k$ of all finite levels
    contains exactly the primitive-recursive functions.
    A function is called \emph{Ackermannian} if it is in
    $\mathfrak{F}_{\omega}\setminus\bigcup_{k\in\N}\mathfrak{F}_k$.
\end{definition}

A sequence $x_0, x_1,\dots,x_l$ of macrostates
is called \emph{good} if there are indices $0\le i<j\le l$
such that $x_i\sqsubseteq x_j$ and \emph{bad} otherwise.
The sequence is \emph{$t$-controlled} by $f:\N\to\N$
if $\norminf{x_i} < f(i+t)$ for every index $0\le i\le l$.
\begin{theorem}[\cite{FFSS2011}]\label{thm:badseq}
    Let $f:\N\to\N$ be a monotone function in $\mathfrak{F}_\gamma$
    such that $f(x)\ge\max\{1,x\}$ for some $\gamma\ge 1$.
    There is a function $L_{k,f}(t)$ in $\mathfrak{F}_{k+\gamma-1}$
    that computes a bound on the maximal length of
    bad sequences in $\Nb^k$ that are $t$-controlled by $f$.
\end{theorem}

\begin{corollary}
    \Tuniv\ of \OCN\ is Ackermannian.
\end{corollary}
\begin{proof}
    By \cref{thm:univ-ack}, it suffices to show that the problem is
    in $\mathfrak{F}_\omega$.
    Recall the procedure that, for a given process $pm$ of a net with
    $k$ control-states, guesses a shortest terminating path from the initial
    macrostate 
    (a witness for non-universality),
    and stops unsuccessfully if a macrostate covers
    one that has been seen before.
    The time and space requirements of this procedure
    are bounded in terms of the longest non-increasing (w.r.t.~covering) sequence
    of $k$-dimensional macrostates.
    These are bad sequences where the norm of the initial macrostate is $m$, the
    counter-value of the process to check for universality.
By point~\ref{lem:macroprops:norminc} of \cref{lem:macroprops}, such sequences
are $m$-controlled by the successor function $f(x)=x+1$, which is in $\mathfrak{F}_1$.
By \cref{thm:badseq},
computing the bound and running the procedure above
is in $\mathfrak{F}_k$.
As $k$ is part of the input, this yields a procedure in $\mathfrak{F}_\omega$.
\qed
\end{proof}

    \section{Conclusion}
    We have shown NL-completeness of the general \tinc\ problem for
deterministic \textOCNs, where initial counter-values are part of the input.
%
Our proof is based on a characterization of the shape of possible witnesses
in terms of a small number of polynomially-sized templates.
Realizability of such templates can be verified in space
logarithmic only in the size of the underlying state space.
Our procedure is therefore independent of the number of action symbols and transitions in the input nets.
To prove the characterization theorem we use witness rewriting rules,
the correctness of which crucially depends on the monotonicity of \tinc\ w.r.t.~counter-values.
In fact, we only make use of this property in the net on the left but
similarly one can define rules that exploit only the monotonicity
in the process on the right.
With some additional effort one can extend this argument
also for trace inclusion between \DOCN\ and \DOCA\ or vice versa
(see \cite{thesis}).

The second part of the paper explores the complexity of the
universality problem for nondeterministic \OCN,
and \tinc\ between finite systems and \OCN\ that easily reduces to
\OCN\ universality.
Here we show that the simplest known algorithm which uses
a well-quasi-order based saturation technique has already optimal complexity:
The problem is Ackermannian, i.e., not primitive recursive.

    \paragraph*{Acknowledgement.}
    We thank Mary Cryan, Diego Figueira and Sylvain Schmitz
for helpful discussions and the anonymous reviewers of an earlier draft for their constructive feedback.
Piotr Hofman acknowledges a partial support by the Polish NCN grant 2013/09/B/ST6/01575.

    \printbibliography[heading=bibintoc]

    \newpage
    \appendix

    \section{Normal-Form Assumption}
    \label{app:reduction}
    We consider here what is sometimes called \emph{realtime}
automata, in which no silent ($\varepsilon$-labelled) transitions are present.
In the absence of zero-tests, the usual syntactic restriction
for deterministic Pushdown Automata, 
(no state with outgoing $\varepsilon$-transition
may have outgoing transitions labelled by $a\neq\varepsilon$)
and the lack of an explicit zero-test in our model
implies that all states on $\varepsilon$-cycles are essentially deadlocks.
A process in such a state can either silently exhaust the counter
and deadlock or divert into an infinite $\varepsilon$ loop.
With respect to their traces, those processes are equivalent.
This means one can eliminate 
$\varepsilon$-transitions by removing $\varepsilon$-cycles
and replacing the remaining short paths by direct steps
(and normalize the effects of single transitions back to $\{-1,0,1\}$).
Such a reduction works in $\O(\log n)$ space.
Allowing $\varepsilon$-transitions thus leaves the complexity of \tinc\ invariant.

\lemreduction*

\begin{proof}
    Let $\AN{A}=(N,\Act,\delta_A)$ and 
    $\AN{B}=(M,\Act',\delta_B)$.
    If $\AN{A}$ is not already deterministic,
    we can make it so by uniquely re-labeling all its transitions $t$ by
    actions $a_t$ and 
    adding corresponding transitions $(p',a_t,d',q')$ to the other net $\AN{B}$
    for any existing $(p',a,d',q')\in\delta_B$, where $a$ is the original label
    of $t\in\delta_A$.
    So assume $\AN{A}$ is deterministic and pick a new action label
    $\$\not\in Act$.
   We
    add
    $\$$-labelled cycles with effect $0$ to all states of $\NN{A}$:
    The new net $\AN{A'}=(N,\Act\cup\{\$\},\ol{\delta_A})$ has transitions
    $\ol{\delta_A} = \delta_A \cup\{(s,\$,0,s)|s\in N\}$.
    To compensate this, we add $\$$-cycles to all states of $\AN{B}$ in the same way.
    We add a sink state $L$ (for losing),
    which has counter-decreasing cycles for all actions,
    and connect all states without outgoing
    $a$-transitions to $L$ by $a$-labelled transitions.
    $\AN{B'}=(M\cup\{L\}, \Act\cup\{\$\},\ol{\delta_B}{})$ where
    \begin{align*}
        \ol{\delta_B}{} = \delta_B &\cup \{(s,\$,0,s)\;|\;s\in M\}\\
        &\cup \{(s,a,0,L)\;|\; a\in\Act\ and\ s\step{a} s'\not\in \delta \ for\ any\ s'\in M\ ) \}\\
        &\cup \{(L,a,-1,L\;|\;a\in \Act\cup\{\$\}\}.
    \end{align*}
    We see that if a word $w$ of length $k$
    witnesses non-inclusion
    $T_\NN{A}(qn)\not\subseteq T_\NN{B}(q'n')$
    then there is a word that witnesses non-inclusion
    $T_\NN{A'}(qn)\not\subseteq T_\NN{B'}(q'n')$
    To see this, 
    observe that in this case,
    any $w$-labelled path
    in $\AN{B'}$ that starts in state $q'$ 
    must end in state $L$.
    This means any such path takes the initial process $q'n'$ to
    some process $Ln''$ where $n''\le n'+k$ and now by playing $n''$ times a label $\$$ we get a new witness.
    Conversely, if 
    there is a witness $w$ for $T_\NN{A'}(qn)\not\subseteq T_\NN{B'}(q'n')$
    then the shortest such witness
    must be of the form $w=w'\$^k$ where $w'$ does not contain actions $\$$
    because as $\$$-labelled steps leave any process not in state $L$ unchanged.
    This means $w'$ witnesses $T_\NN{A}(qn)\not\subseteq T_\NN{B}(q'n')$.
    \qed
\end{proof}

    \section{Checking Weighted Inequalities in Logspace}
    \label{app:inequalities}
\begin{lemma}\label{lem:complexity}
    Inequalities
    of the form $m \cdot A + B \ge n\cdot C + D$
    where all coefficients are non-negative integers given in binary
    can be verified in $\O(\log (A+B+C+D))$ deterministic space.
\end{lemma}
\begin{proof}
Assume w.l.o.g.~that the bit-representations of $m$ and $n$ are of the same length,
as are those of $A,B,C$ and $D$, and we have the least significant bit on the right.

To check $m\ge n$, we can stepwise read their binary representation
from right to left, flipping
an ``output'' bit $\textit{Out}$ on the way:
Initially, $\textit{Out}:=1$; 
in every step set $\textit{Out}:=0$ if the current bit in $m$
is strictly smaller than that in $n$; set
$\textit{Out}:=1$ if the current bit in
$m$ is strictly bigger than that in $n$
and otherwise proceed without touching $\textit{Out}$.
The inequality holds iff
$\textit{Out}=1$ after completely reading the input.

To check the weighted variant, 
we use the same algorithm but multiply $m\cdot A$, and $n\cdot C$ on the fly,
using standard long binary multiplication.
We use a scratchpad to store the intermediate sums, starting with
values $B$ and $D$.
In a step that reads the $i$th bit $m[i]$ of $m$,
we want to add $A\cdot 2^i$ to the intermediate sum if $m[i]=1$.
We can do that by shifting the binary representation of $A$
left $i$ times and adding the result to the current scratchpad.
We see that none of the bits up to $i-1$ in the scratchpad
are affected by this operation. We can therefore
discard (and use for the comparison in our simple algorithm above)
the rightmost bit of the scratchpad in every step.
The claim now follows from the observation
that the necessary size of the scratchpad is bounded by $B+A+1$.
\qed
\end{proof}

    \newpage
     \section{Proof of Theorem \ref{thm:form}}
    \label{sec:form_proof}
    We show that it is safe to consider only witnesses in a reduced form,
and derive bounds on the length of certain subpaths.
For this, we introduce path rewriting rules that exchange
occurrences of some loops by others.
We then show (in \cref{lem:witness_preservation}) that these rules preserve
witnesses and (\cref{lem:rules_wqo}) cannot be applied
indefinitely.
For \emph{reduced} witnesses, those to which no rules are applicable,
we derive (\cref{lem:bounds}) bounds on the
multiplicities of loops that are less productive than others,
which will enable us to prove \cref{thm:form}.

For the rest of this section let $V$ and $E$ be the sets of nodes and edges in the product of
$\NN{A}$ and $\NN{B}$.

We start with an easy observation:
Because no loop $L$ is longer than $|V|$, we conclude that $(\effect{\SP}(L),\effect{\DUP}(L))\in \{-V\ldots V\}\times \{-V\ldots V\},$ so there are $F_0 := (2\cdot|V|+1)^2$ 
different values the pair $\effect{\SP}(L),\effect{\DUP}(L)$
can have.
Moreover, if a witness exists, then also one that
does not contain different loops with the
same effects: if
$\pi_0L_0\pi_1L_1\pi_2$ is a witness where $|\pi_1|>0$ and $L_0,L_1$ are two loops with
$\effect{}(L_0)=\effect{}(L_1)$,
then either some prefix of $\pi_0L_0^2\pi_1\pi_2$ (if $\effect{A}(L_0)\ge 0$)
or some prefix of $\pi_0\pi_1L_1^2\pi_2$
(if $\effect{A}(L_0)<0$) 
must also be a witness by \cref{lem:witness-monotonicity}.
We can therefore consider w.l.o.g.~only \emph{sane} paths,
which are of the form
    \begin{equation}\label{eq:form:sane}
        \pi=\pi_0L_0^{l_0}\pi_{1}L_{1}^{l_1}\dots\pi_rL_r^{l_r}\pi_{r+1}
    \end{equation}
where $r\le F_0$, all $\pi_i$ are acyclic and all loops have pairwise different effects.

\begin{definition}[Path Rewriting Rules]
    \label{def:rules}
  Consider the rules given
  below. 
  
  \begin{center}
\newcommand{\RuleNameW}{1.2cm}
\newcommand{\RuleW}{3.74cm}

\Ruledef{\RuleNameW}{\RulePL}{\RuleW}{
    $\pi=\pi_0L_0^{l_0}\pi_1L_1^{l_1}\pi_2$\\
    $Type (L_0) = \TypeUU$\\
    $Type (L_1) = \TypeUU$\\
    $\Delta_\DUP(L_0) \cdot x = \Delta_\DUP(L_1) \cdot y$\\
    $S(L_0)\ge S(L_1)$\\
    $l_1-y>0$
}{
    $\rho=\pi_0L_0^{l_0+x}\pi_1L_1^{l_1-y}\pi_2$
}
\hspace{0.5cm} 
\Ruledef{\RuleNameW}{\RulePR}{\RuleW}{
    $\pi=\pi_0L_0^{l_0}\pi_1L_1^{l_1}\pi_2$\\
    $Type (L_0) = \TypeUU$\\
    $Type (L_1) = \TypeUU$\\
    $\Delta_\DUP(L_0) \cdot x = \Delta_\DUP(L_1) \cdot y$\\
    $S(L_0)< S(L_1)$\\
    $l_0-x>|\pi_1L_1|$
}{
    $\rho=\pi_0L_0^{l_0-x}\pi_1L_1^{l_1+y}\pi_2$
}

\vspace{0.5cm} 

\Ruledef{\RuleNameW}{\RulePN}{\RuleW}{
    $\pi=\pi_0L_0^{l_0}\pi_1L_1^{l_1}\pi_2$\\
    $Type (L_0) = \TypeUU$\\
    $Type (L_1) = \TypeDD$\\
    $\Delta_\DUP(L_0) \cdot x = - \Delta_\DUP(L_1) \cdot y$\\
    $S(L_0)\le S(L_1)$\\
    $l_0-x\ge |\pi_1|$\\
    $l_1-y>0 \land l_0-x>0$
}{
    $\rho=\pi_0L_0^{l_0-x}\pi_1L_1^{l_1-y}\pi_2$
}
\hspace{0.01cm} 
\Ruledef{\RuleNameW}{\RuleNL}{\RuleW}{
    $\pi=\pi_0L_0^{l_0}\pi_1L_1^{l_1}\pi_2$\\
    $Type (L_0) = \TypeDD$\\
    $Type (L_1) = \TypeDD$\\
    $\Delta_\DUP(L_0) \cdot x = \Delta_\DUP(L_1) \cdot y$\\
    $S(L_0)< S(L_1)$\\
    $l_1>|L_0|\cdot x + 2|\pi_1|$\\
    $l_1-y>0$
}{
    $\rho=\pi_0L_0^{l_0+x}\pi_1L_1^{l_1-y}\pi_2$
}
\hspace{0.01cm} 
\Ruledef{\RuleNameW}{\RuleNR}{\RuleW}{
    $\pi=\pi_0L_0^{l_0}\pi_1L_1^{l_1}\pi_2$\\
    $Type (L_0) = \TypeDD$\\
    $Type (L_1) = \TypeDD$\\
    $\Delta_\DUP(L_0) \cdot x = \Delta_\DUP(L_1) \cdot y$\\
    $S(L_0)\ge S(L_1)$\\
    $l_0-x>0$
}{
    $\rho=\pi_0L_0^{l_0-x}\pi_1L_1^{l_1+y}\pi_2$
}

%
%
    
  \end{center}

  Each rule consists of \emph{conditions} (lines above the bar)
  and a conclusion $\rho$, which is a path, below the bar.
Their names indicate which type of loops are handled:
E.g., \RulePL\ exchanges loops of type $\TypeUU$ (up)
for others of the same type on the left.
  
  We say a rule is \emph{applicable} to a sane path $\pi$
  if there are $0<x,y,l_0,l_1\in\N$ and two different loops $L_0$ and $L_1$
  such that all conditions are satisfied.
  In this case the rule can rewrite $\pi$ to $\rho$,
  its conclusion and we say $\rho$ is the result of applying the rule to $\pi$.
\end{definition}
\begin{figure}[H]
\label{fig:rules}
\end{figure}

\begin{example}
Consider \cref{ex:witness} again:
The substitution suggested there is an application of the rule
\RulePL\ to the
path $\pi=(t_0t_1t_2)(t_3t_4)^9t_5(t_6)^{20}$,
where $L_0=(t_0t_1t_2)$, $L_1=(t_3t_4)$
and $x=y=8$. The result is a
reduced witness for $(p0,p'10)$ of length $50$.
Shorter reduced witnesses for $(p0,p'10)$ exist, for example
$(t_0t_1t_2)^6t_5t_6^{16}$, but because of their different loop structure, these
cannot be obtained from $\pi$ by applying rewriting rules,
as these do not change the structure, i.e., which loops
occur and in which order, of a path.
This means that our rules do not necessarily preserve minimality of witnesses.
\end{example}

In the next two \cref{lem:witness_preservation,lem:rules_wqo},
we show that the rewriting rules preserve witnesses
and that
continuous rule application must eventually terminate.
\begin{restatable}[]{lemma}{lemwitnesspreservation}
\label{lem:witness_preservation}
    If $\pi$ is a sane witness for $(pm,p'm')$
    and $\rho$ is the result of applying one of the rules
    to $\pi$, then $\rho$ is also a sane witness for $(pm,p'm')$.
\end{restatable}
\begin{proof}
    Each rule only modifies the number of times
    some loops are iterated, and never completely removes
    a loop. Therefore, sane paths are always rewritten to other sane paths.

    Let's say we rewrite $\pi=\pi_0L_0^{l_0}\pi_1L_1^{l_1}\pi_2$ to $\rho$.
    The key observation is that the conditions of the rule
    imply that we can always decompose
    the paths $\pi$ and $\rho$ into $\pi=\alpha\gamma$
    and $\rho=\beta\gamma$, s.t.\
    $\deffect{\alpha}=\deffect{\beta}$ and 
    $\seffect{\alpha}\le \seffect{\beta}$.
    By monotonicity (\cref{lem:monotonicity})
    and the assumption that $\pi$ is a witness,
    it is therefore sufficient to show that
    the result $\rho$ is still enabled in the initial position $(pm,p'm')$.
    We proceed by case distinction for the used rule.
%

    \RulePL. 
    Since $\pi$ is a witness, its prefix $\alpha=\pi_0L_0^{l_0}\pi_1L_1^{l_1}$
    must be enabled in $(pm,p'm')$ and because $\Type{L_0}=\TypeUU$,
    so is the prefix $\beta=\pi_0L_0^{l_0+x}\pi_1L_1^{l_1-y}$ of the result
    $\rho$. Assume that $(pm,p'm')\step{\alpha}(qn,q'n')$
    and
    $(pm,p'm')\step{\beta}(q\hat{n},q'n')$.
    The condition $S(L_0)\ge S(L_1)$ of the rule implies that $\hat{n}\ge n\ge\sguard{\pi_2}$
    and therefore that $\rho$ is enabled in $(pm,p'm')$.
%
%
    
    \RulePR.
    The prefix $\pi_0L_0^{l_0-x}$ of $\pi$ must be enabled
    and since the last condition of the rule demands that $l_0-x>|\pi_1L_1|$,
    so is the path $\pi_0L_0^{l_0-x}\pi_1L_1$.
    The fact that $\Type{L_1}=\TypeUU$, means that also
    $\pi_0L_0^{l_0-x}\pi_1L_1^{l_1+y}$ and therefore the result $\rho$
    is enabled in $(pm,p'm')$.

\RulePN.
    $\Type{L_1}=\TypeDD$ implies $S(L_1)<\infty$. Since $S(L_0)<S(L_1)$, 
    we know that $S(L_0)<\infty$ and hence $\deffect{L_0}>0$.
    The path $\pi_0L_0^{l_0-x}$ is a prefix of $\pi$ and is therefore
    enabled in $(pm,p'm')$.
    As $l_0-x\ge |\pi_1|$ by assumption, we get that
    \begin{equation}
        m+\seffect{\pi_0L_0^{l_0-x}}\ge l_0-x\ge|\pi_1| \ge\sguard{\pi_1}
    \end{equation}
    and similarly, by $\deffect{L_0}>0$,
    \begin{equation}
        m'+\deffect{\pi_0L_0^{l_0-x}}\ge l_0-x\ge|\pi_1| \ge\dguard{\pi_1}.
    \end{equation}
    This means that the prefix $\beta=\pi_0L_0^{l_0-x}\pi_1$ of $\rho$ is enabled in $(pm,p'm')$.
    Let us now consider the prefix $\alpha=\pi_0L_0^{l_0-x}L_0^x\pi_1L_1^{y}$ of
    $\pi$.
    Because $\deffect{L_0}\cdot x = - \deffect{L_1}\cdot y$ we get
    $\deffect{\alpha} = \deffect{\beta}$.
    By $S(L_0)<S(L_1)$ we obtain that
    $\seffect{\alpha} \le \seffect{\beta}$.
    Because $\pi=\alpha L_1^{l_1-y}\pi_2$ is a witness for $(pm,p'm')$,
    we can apply \cref{lem:monotonicity} to conclude
    $\rho=\beta L_1^{l_1-y}\pi_2$ must be a witness for $(pm,p'm')$.

    \RuleNL.
    We know that $m+\seffect{\pi_0L_0^{l_0}}+\seffect{\pi_1} \ge \sguard{L_1^{l_1}}$,
    because $\pi$ is enabled in $(pm,p'm')$.
    As $L_1$ is a type $\TypeDD$ loop we also know that $\seffect{L_1}<0$.
    Therefore, $\sguard{L_1^{l_1}}\ge l_1$ and
    \begin{equation}
        \label{eq:shift-negative-left:2}
        m+\seffect{\pi_0L_0^{l_0}} \ge l_1 -\seffect{\pi_1}.
    \end{equation}
    Assume towards a contradiction that
    $m+\seffect{\pi_0L_0^{l_0}} < \sguard{L_0^x\pi_1}$.
    This means that
    \begin{equation}
        m+\seffect{\pi_0L_0^{l_0}} <\sguard{L_0^x} +|\pi_1| \le |L_0|\cdot x +|\pi_1|.
    \end{equation}
    This, together with \cref{eq:shift-negative-left:2} yields $l_1-\seffect{\pi_1}<|L_0|\cdot x +|\pi_1|$
    and thus $l_1<|L_0|\cdot x +2|\pi_1|$
    which contradicts the condition that $l_1>|L_0|\cdot x +2|\pi_1|$.
    Hence, $ m+\seffect{\pi_0L_0^{l_0}} \ge \sguard{L_0^x\pi_1}$.
    By the same argument we get that $m'+\deffect{\pi_0L_0^{l_0}}\ge\dguard{L_0^x\pi_1}$.
    So the prefix $\beta=\pi_0L_0^{l_0+x}\pi_1$ of $\rho$
    is enabled in $(pm,p'm')$.
    Consider the prefix $\alpha=\pi_0L_0^{l_0}\pi_1L_1^y$ of $\pi$.
    By the assumption that $\deffect{L_0^x} = \deffect{L_1^y}$ we get
    that $\deffect{\alpha}=\deffect{\beta}$.
    Because of $S(L_0)<S(L_1)$
    we get $\seffect{L_0^x} \ge \seffect{L_1^y}$ and therefore
    that 
    $\seffect{\alpha}\le\seffect{\beta}$.
    By \cref{lem:monotonicity} we conclude that
    the path $\rho=\beta L_1^{l_1-y}\pi_2$ is a witness for $(pm,p'm')$.
    
    \RuleNR. 
    Let $\alpha=\pi_0L_0^{l_0}\pi_1$ and let
    $(pm,p'm')\step{\alpha}(qn,q'n')$.
    Due to the type of $L_0$ and because $\pi$ is a witness,
    we know that the prefix
    $\beta=\pi_0L_0^{l_0-x}\pi_1L_1^{y}$ of $\rho$
    is enabled in $(pm,p'm')$.
    Since $\deffect{L_0}\cdot x=\deffect{L_1}\cdot y$,
    we get that $(pm,p'm')\step{\beta}(q\hat{n},q'n')$
    for some $\hat{n}\in\N$.
    The condition
    $S(L_0)\ge S(L_1)$ of the rule
    implies that $\seffect{L_0^x}\le \seffect{L_1^y}<0$,
    and therefore that $\hat{n}\ge n$.
%
    We conclude that the path $L_1^{l_1}\pi_2$
    is enabled in $(qr,q'r')$
    and therefore that $\rho=\pi_0L_0^{l_0-x}\pi_1L_1^{l_1+y}\pi_2$
    is enabled in $(pm,p'm')$ as required.

\qed
\end{proof}


\begin{restatable}[]{lemma}{lemruleswqo}
\label{lem:rules_wqo}
  Any sequence of successive applications of rules to a given
  path $\pi$ must eventually terminate.
\end{restatable}
\begin{proof}
Consider a $\pi$ to wich we apply the rewriting rules. W.l.o.g.~assume
$\pi$ is sane, as otherwise no rule is applicable by definition.
The \label{def:decompositions}
\emph{decomposition} of $\pi$ is the sequence
\begin{equation}
    \label{eq:witness-decomposition}
    Dec(\pi) =
    (\pi_0,L_0,l_0)
    (\pi_1,L_1,l_1)
    \dots
    (\pi_k,L_k,l_k)
    \pi_{k+1}
\end{equation}
in $(E^*\x E^*\x\N)^*E^*$
such that $\pi=\pi_0L_0^{l_0}\pi_1L_1^{l_1}\dots\pi_kL_k^{l_k}\pi_{k+1}$, where $k\le F_0$ and
for all indices $0\le i \le k$,
\begin{enumerate}
    \item $L_i$ is a loop,
    \item $\pi_i$ is acyclic, 
    \item for any two transitions $t\in \pi_i$ and $t'\in L_i$ with $target(t)=target(t')$
        holds that $target(L_i)=target(t)$.
\end{enumerate}
The last condition demands that any loop $L_i$ shares exactly one node with the
acyclic path $\pi_i$ it succeeds and thus ensures that the decomposition
of a path is unique.

As no application of a rule completely removes all occurrences of loops nor
introduces new ones nor touches the intermediate paths, we observe that
rule applications only change the exponents $l_i$ in the decomposition of the path.

Based on the order of loops in the decomposition of $\pi$, and their potential
for rule application, we now define a notion of \emph{weights} for paths,
and show that these weights have to strictly decrease along a well-order whenever
a rule is applied.

Let $(L_0,L_1,\dots,L_k)$ be the sequence of loops that occur
in the decomposition of $\pi$.
Let us fix some linear order $\prec$ on $\{L_0,L_1,\dots,L_k\}$ that satisfies the following
conditions for any two different loops $L_i,L_j$ with $i<j$.
\begin{enumerate}
    \item\label{cond:PL} If $Type(L_i)=Type(L_j)=\TypeUU$ and $S(L_i)\ge S(L_j)$ then $L_i\prec L_j$.
    \item\label{cond:PR} If $Type(L_i)=Type(L_j)=\TypeUU$ and $S(L_i)<S(L_j)$ then $L_i\succ L_j$.
    \item\label{cond:NL} If $Type(L_i)=Type(L_j)=\TypeDD$ and $S(L_i)<S(L_j)$ then $L_i\prec L_j$.
    \item\label{cond:NR} If $Type(L_i)=Type(L_j)=\TypeDD$ and $S(L_i)\ge S(L_j)$ then $L_i\succ L_j$.
\end{enumerate}
\newcommand{\SOL}{\sigma}
Surely, such a linearization exists, as the conditions above
only restrict $\prec$ between loops of the same type
and slopes are linearly ordered.
Consider the permutation $\SOL:\N_{\le k}\to\N_{\le k}$
given by $\SOL(i)<\SOL(j)\iff L_i\prec L_j$.
The \emph{weight} of $\pi$ is
\begin{equation}
    W(\pi)=(l_{\sigma(k)},l_{\sigma(k-1)},\dots,l_{\sigma(0)}) \in \N^{k+1}.
\end{equation}
The weight of $\pi$ is the ordered tuple of exponents $l_i$ of loops
that occur in $\pi$. Because rules do not change the order
of loop occurrences, the path before and after applying a rule
have comparable weights. The very definition of weights ensures
that rule applications must strictly reduce the weight of a path.
%
\begin{claim}\label{lem:wellordering}
    If $\rho$ is the result of applying a rewriting rule
    to $\pi$ then $W(\rho) \sqsubset_{lex} W(\pi)$ where
    $\sqsubset_{lex}$ is the lexicographic extension of 
    the pointwise ordering of tuples of naturals.
\end{claim}
    Assume the decompositions of $\pi,\rho$ are
    \begin{equation}
    \begin{aligned}
      Dec(\pi) &= (\pi_0,L_0,l_0)(\pi_1,L_1,l_1)\dots(\pi_k,L_k,l_k)\pi_{k+1} \text{ and}\\
      Dec(\rho) &= (\pi_0,L_0,l_0')(\pi_1,L_1,l_1')\dots(\pi_k,L_k,l_k')\pi_{k+1}.
    \end{aligned}
    \end{equation}
    %
    We show for every type of rule that if the occurrences
    of loop $L_i$ increase then those of some loop $L_j$ with
    $L_i\prec L_j$ strictly decrease.

    If the rule used to derive $\rho$ was \RulePL\ then
    $l_i'=l_i+x$ and 
    $l_j'=l_j-y$ for some $i<j$, $0<x,y$ and type $\TypeUU$ loops
    $L_i,L_j$ with $S(L_i)\ge S(L_j)$.
    By condition \ref{cond:PL}) in the definition of $\prec$
    we get $L_i\prec L_j$. 

    For rule \RulePR\ we know $l_i'=l_i-x$ and 
    $l_j'=l_j+y$ for some $0<x,y$ and type $\TypeUU$ loops
    $L_i,L_j$ with $S(L_i)<S(L_j)$.
    By condition \ref{cond:PR}) in the definition of $\prec$,
    we get $L_i\succ L_j$.

    For rule \RuleNL\ we know $l_i'=l_i+x$ and 
    $l_j'=l_j-y$ for type $\TypeDD$ loops
    $L_i,L_j$ with $S(L_i)<S(L_j)$.
    By condition \ref{cond:NL}) in the definition of $\prec$,
    we know $L_i\prec L_j$.
    %

    For rule \RuleNR\ we know $l_i'=l_i-x$ and 
    $l_j'=l_j+y$ for some $0<x,y$ and type $\TypeDD$ loops
    $L_i,L_j$ with $S(L_i)>S(L_j)$.
    So condition \ref{cond:NR}) in the definition of $\prec$,
    implies $L_i\succ L_j$.
    %

    Lastly, if the rule used to derive $\rho$ was \RulePN\ we immediately see that $l_i'<l_i$ and
    $l_j'<l_j$, which implies the claim.

%
\qed
\end{proof}

\Cref{lem:witness_preservation,lem:rules_wqo} allow us to focus on witnesses that are \emph{reduced}, i.e.,
which are sane and to which none of the rewriting rules is applicable. 
We can now derive bounds on the multiplicities of loops
in reduced paths.

\begin{restatable}[]{lemma}{lembounds}
\label{lem:bounds}
    Let
    $\pi=\pi_0L_0^{l_0}\pi_1L_1^{l_1}\pi_2$
    be a reduced path
    where
    $L_0,L_1$ are loops occurring with multiplicities
    $l_0>0$ and $l_1>0$. 
    \begin{enumerate}
        \item\label{lem:bounds:PL}
            If $Type(L_0)=Type(L_1)=\TypeUU$ and $S(L_0)\ge S(L_1)$
            then $l_1\le |V|$
        \item\label{lem:bounds:PR}
            If $Type(L_0)=Type(L_1)=\TypeUU$ and $S(L_0) < S(L_1)$
            then $l_0\le |\pi_1| + 2|V|$
        \item\label{lem:bounds:NL}
            If $Type(L_0)=Type(L_1)=\TypeDD$ and $S(L_0)<S(L_1)$
            then $l_1< |V|^2 + 2|\pi_1|$
        \item\label{lem:bounds:NR}
            If $Type(L_0)=Type(L_1)=\TypeDD$ and $S(L_0)\ge S(L_1)$
            then $l_0< |V|$
        \item\label{lem:bounds:PN}
            If $Type(L_0)=\TypeUU$, $Type(L_1)=\TypeDD$ and $S(L_0) \le S(L_1)$
            then $l_0\le |\pi_1|+|V|$ or $l_1\le |V|$.
    \end{enumerate}
\end{restatable}
\begin{proof}
    The fourth condition of any rule
    is satified e.g.~by $x=\effect{\DUP}(L_1)$ and $y=\effect{\DUP}(L_0)$.
    So if $0<x,y\in\N$
    is the smallest satisfying pair we know $x,y\le |V|$.
    The bounds are now easily derived by contradiction:
    \begin{enumerate}
      \item If $l_1\ge |V|$ then $l_1-y\ge l_1-|V|>0$ 
          and rule \RulePL\ is applicable.
      \item If $l_0> |\pi_1| +2|V|$ then $l_0-x > |\pi_1| + 2|V| - x \ge |\pi_1| + |L_1| \ge |\pi_1L_1|$
          and therefore rule \RulePR\ is applicable.
      \item If $l_1\ge |V|^2+2|\pi_1|$ then $l_1\ge |L_0|\cdot x + 2|\pi_1|$
          and $l_1-y\ge l_1-|V|>0$, so rule \RuleNL\ is applicable.
      \item If $l_0>|V|$ then $l_0-x>0$, so rule \RuleNR\ is applicable.
      \item If $l_1>|V|$ and $l_0>|\pi_1|+|V|$, then $l_1-y>0$, $l_0-x>0$ and $l_0-x>|\pi_1|$,
          so rule \RulePN\ is applicable.
    \end{enumerate}
%
%
%
%
    In each case we conclude that one of the rules 
    is applicable,
    which contradicts the assumption that $\pi$ is reduced.
    \qed
\end{proof}

%
Finally, we are ready to prove \cref{thm:form}.
\thmform*
\begin{proof}
    We show that we can sufficiently increase the bound $\fbound$
    such that whenever $T(pm)\not\subseteq T(qn)$ but no witness
    exists that is shorter than $\fbound$ or of forms 1) or 2),
    then there must be a witness of form 3).

    Assume $T(pm)\not\subseteq T(qn)$ and consider a reduced witness $\pi$,
    that is minimal in length: no shorter witness is reduced.
    Recall that this also means that $\pi$ is sane: it is of form
    described in \cref{eq:form:sane}.
    By monotonicity (\cref{lem:witness-monotonicity}) and because $\pi$ is of minimal length
    among the reduced witnesses, we see that it cannot contain loops of type $\TypeDU$.
    Moreover, because $\pi$ is not of form 1),  we can safely assume that $\pi$
    it contains
    only loops of types $\TypeUU$ and $\TypeDD$.
    This is because if a witness contains two or more different type $\TypeUD$
    loops, then there exists another (sane) witness, that only unfolds the first
    such loop.
    Relaxing the bound on the length of paths between loops to
    $F_1:=F_0(2|V|+|V|^2)$, we can write $\pi$ as
    \begin{equation}
        \label{eq:form}
        \pi=\pi_0L_0^{l_0}\pi_{1}L_{1}^{l_1}\dots\pi_kL_k^{l_k}\pi_{k+1}
    \end{equation}
    where $k\le F_0$, all 
    $|\pi_i| < F_1$ and the number of iterations of loop $L_i$ is $l_i> |V|$.

    Consider a block $\pi_{\mathit{pos}}=L_i^{l_i}\pi_{i+1}L_{i+1}^{l_{i+1}}\pi_{i+2}\dots\pi_jL_j^{l_j}$
    that is part of the decomposition above, 
    %
    such that all loops are type $\TypeUU$.
    If
    for indices $i\le x<y\le j$ we have
    $S(L_x) \ge S(L_y)$,
    then by
    \cref{lem:bounds}.\ref{lem:bounds:PL} we get $l_y \le |V|$.
    Therefore, $\pi_{\mathit{pos}}$ can be rewritten to the form
    \begin{equation}
        \pi_{\mathit{pos}}=L_i^{l_i}\pi_{i+1}L_{i+1}^{l_{i+1}}\pi_{i+2}\dots\pi_jL_j^{l_j}\pi_{j+1}
    \end{equation}
    where the lengths of $\pi_i$ are bounded by $F_2 := F_0\cdot(|V|^2 + F_1)$ and
    the slopes of loops are strictly increasing:
    $S(L_x) < S(L_y)$ for any two indices $i\le x<y\le j$.
    By \cref{lem:bounds}.\ref{lem:bounds:PR}
    this means that
    $l_x\le |\pi_{x+1}|+2|V| \le F_2 + 2|V| =: F_3$.
    We conclude that the prefix
    $\pi'=L_i^{l_i}\pi_{i+1}L_{i+1}^{l_{i+1}}\pi_{i+2}\dots\pi_{j-1}L_{j-1}^{l_{j-1}}$
    is no longer than $(j-i)\cdot (|V|\cdot F_3+ F_2)$
    and therefore
    \begin{equation}
        \label{eq:squeezing-form-pos}
        \pi_{\mathit{pos}} = \pi'L_j^{l_j}\pi_{j+1}
    \end{equation}
    where $|\pi'|$ is bounded by $F_4:= F_0(|V|\cdot F_3+ F_2)$
    and $|\pi_{j+1}|$ by $F_2$.

    We continue to show by a similar argument
    that we can bound the number of iterations of all but the most productive loop
    in a block consisting of only decreasing (type $\TypeDD$) loops.
    Consider a block $\pi_{\mathit{neg}}=L_i^{l_i}\pi_{i+1}L_{i+1}^{l_{i+1}}\pi_{i+2}\dots\pi_jL_j^{l_j}$
    that is part of the decomposition in \cref{eq:form},
    where all loops are type $\TypeDD$.
    %
    If
    $S(L_x) \ge S(L_y)$
    for some indices $i\le x<y\le j$, then by
    \cref{lem:bounds}.\ref{lem:bounds:NR} we know $l_y < |V|$.
    This means that $\pi_{\mathit{neg}}$ is of the form
    \begin{equation}
        \pi_{\mathit{neg}}=\pi_iL_i^{l_i}\pi_{i+1}L_{i+1}^{l_{i+1}}\pi_{i+2}\dots\pi_jL_j^{l_j}\pi_{j+1}
    \end{equation}
    where all $\pi_i$ have lengths bounded by $F_2$ and
    $S(L_x) < S(L_y)$ for any two indices $i\le x<y\le j$.
    By \cref{lem:bounds}.\ref{lem:bounds:NL}
    we get $l_y\le |V|^2 + 2|\pi_{x}| \le |V|^2 + 2F_2=: F_3'$ and
    conclude that the suffix
    $\pi''=\pi_{i+1}L_{i+1}^{l_{i+1}}\pi_{i+2}\dots\pi_{j}L_{j}^{l_{j}}\pi_{j+1}$
    is no longer than $(j-i)\cdot (|V|\cdot F_3'+F_2)$.
    Therefore, $\pi_{\mathit{neg}}$ is of the form
    \begin{equation}
        \label{eq:squeezing-form-neg}
        \pi_{\mathit{neg}} = \pi_iL_i^{l_i}\pi''
    \end{equation}
    where $\pi_i$ is bounded by $F_2$ and $\pi''$ by $F_4':= F_0(|V|\cdot F_3'+ F_2)$.

    \cref{eq:squeezing-form-pos,eq:squeezing-form-neg} characterize the
    form of maximal subpaths of the witness $\pi$ in \cref{eq:form},
    along which the type of loops does not change.
    They allow us to write $\pi$ as
    \begin{equation}\label{eq:form:alt}
        \pi=\pi_0L_0^{l_0}\pi_1L_{1}^{l_1}\pi_2\dots\pi_kL_k^{l_k}\pi_{k+1}
    \end{equation}
    where for all indices $0\le i<k$: 
    \begin{enumerate}
        \item $\pi_i$ is no longer than $F_5:=F_3+F_3'+F_4+F_4'$.
        \item All $l_i> |V|$.
        \item Consecutive loops $L_i$ and $L_{i+1}$ have different types.
        \item If loops $L_i,L_j$ for $0\le i<j\le k$ have the same type then
            $S(L_i)<S(L_j)$.
    \end{enumerate}
    In the remainder of this proof, we further increase the polynomial bound for
    the gaps $\pi_i$ between the loops;
    this allows to conclude that $\pi$ contains at least one type $\TypeDD$ loop
    and finally, that $\pi$ is of form 3).

    Observe that if all loops $L_i$ in \cref{eq:form:alt} are of type
    $\TypeUU$ then the witness is already of form
    $\pi=\pi_0L^l\pi_1$
    as in \cref{eq:squeezing-form-pos},
    where $\pi_0,\pi_1$ are short and $L$ is the most effective loop.
    In this case,
    consider the run
    \begin{equation}
        (pm,qn)\step{\pi_0L^{l}}(p'm',q'n')
    \end{equation}
    induced by the prefix $\pi_0L^l$.
    Because $\NN{B}$ is complete we know
    $\effect{\DUP}(\pi)=-n$.
    together with $\effect{\DUP}(\pi_1) \le |\pi_1|\le F_5$
    we get $n'\le F_5$.
    Because $\guard{\SP}(\pi_1)\le|\pi_1|$,
    we know that $l\le |\pi_1|\le F_5$
    as otherwise, fewer iterations $l$ would result in
    a shorter witness and we assumed $\pi$ to be minimal in length.
    Hence, we could bound $\pi$ by $F_6:= F_5+|V|\cdot F_5 + F_5$.
    So if we let $\fbound\ge F_6$, our witness $\pi$ must contain type $\TypeDD$ loops
    as it is assumed not to be no shorter than $\fbound$.

    Finally, fix an index $0\le x \le k$ such that in \cref{eq:form:alt},
    $L_x$ is a loop of type $\TypeDD$ with most efficient decrease (minimal slope). 
    That is, $\pi$ is of form
    \begin{equation}\label{eq:form:4}
      \pi=\pi_0L_x^{l_x}\pi_1.
    \end{equation}
    We now bound both $\pi_0$ and $\pi_1$ and thereby prove that $\pi$ is of form 3).
    We start with the suffix $\pi_1$.

    If $L_x$ is the only loop of type $\TypeDD$, we are done as then $|\pi_1|\le F_5$.
    Suppose we have two indices $0\le y<y+2\le k$, where both $L_y$ and $L_{y+2}$ are type $\TypeDD$.
    This means that $L_{y+1}$ is of type $\TypeUU$ with $S(L_{y+1}) < S(L_{y+2})$.
    By \cref{lem:bounds}.\ref{lem:bounds:PN} and the
    fact that $l_{y+2}> |V|$ we know that
    $l_{y+1}<|\pi_{y+1}|+ |V|\le F_6$.
    So $\pi_{y+1}L_{y+1}^{l_{y+1}}\pi_{y+2}$ is no longer than
    $2\cdot F_5+|V|\cdot F_6=:F_7$.
    Applying \cref{lem:bounds}.\ref{lem:bounds:NL}
    to $L_y$ and $L_{y+2}$ we get
    $l_{y+2}\le|V|^2+2\cdot F_7=:F_8$
    and thus $\pi_{y+1}L_{y+1}^{l_{y+1}}\pi_{y+2}L_{y+2}^{l_{y+2}}$ is no longer than
    $F_9:= F_5 + (|V|\cdot F_6) + F_5 + (|V|\cdot F_8)$.
    Now the above argument can be repeated for any successive pair of type $\TypeDD$ loops in
    $\pi_1$ of which there are at most $F_0$. So,
    $|\pi_1|<F_0\cdot F_9$.

    To bound the prefix $\pi_0$ in \cref{eq:form:4},
    we recall (point 3 above) that consecutive loops in \cref{eq:form:alt}
    have different types and therefore $x\leq 1$.
    In case $x=0$, we immediately get $|\pi_0| \le F_5$.
    If $x=1$, then $L_0$ is a type $\TypeUU$ loop with
    $S(L_0) < S(L_x)$ and so
    by \cref{lem:bounds}.\ref{lem:bounds:PN}
    and point 2), we get $l_0\le |\pi_1|+|V|< F_6$.
    This means $|\pi_0|\le 2F_5+ |V|\cdot F_6=F_7$.

    We conclude that $\fbound:=F_9\cdot F_0$ is sufficient to ensure that any
    witness $\pi$, longer than $\fbound$ which is not of form 1) or 2) must have form 3).
    This completes our argument for the existence of witnesses in the claimed
    forms.

    To see why $l_0$ and $l_1$ can always be bounded polynomially
    in $|V|$ and $m'$ can be seen by looking at the types of the loops involved.
    For paths of form 1 and 3, $L_0$ decreases
    the counter on the right at least once in every iteration.
    Since the value $m'+\deffect{\pi_0}$ before the first iteration is at most
    $m'+c$, we have $l_0\le m'+c$.

    Paths of the second form can be decomposed into a prefix $\pi_0L_0^{l_0}$
    and a suffix $\pi_1L_1^{l_1}\pi_2$, which is a path of form 3.
    Let $y_0\in\N$ be minimal such that the effect
    of the path $\gamma_0=\pi_0L_0^0\pi_1L_1^{y_0}\pi_2$,
    in which $L_0$ is not
    iterated at all is sufficient to reduce the initial value $m'$ below $0$.
    That is, we have $m'+\deffect{\pi_0L_0^0\pi_1L_1^{y_0}\pi_2}\le 0$.
    Note that as for forms 1 and 3, we can bound $y_0$ by $m'+2c$
    and therefore, $|\gamma_0|$ is no larger than $3c +|V|\cdot (m'+2c)$.
    This path might not be a witness because
    it is not enabled on the left side.
    However, because of the condition on the slopes, there are $x,y\le |V|$ such
    that the effect of the loops satisfy
    \begin{align}
        \deffect{L_0}\cdot x = - \deffect{L_1}\cdot y
        \qquad\text{and}\qquad
        \seffect{L_0}\cdot x > - \seffect{L_1}\cdot y.
    \end{align}
    This means, increasing the iterations of the loops $L_0$ and $L_1$
    by $x$ and $y$, respectively,
    does not change the effect of the path on the right,
    but strictly increases the effect on the left.
    We increase the iterations $(l_0,l_1)=(0,y_0)$ in $\gamma_0$
    as suggested above for
    $\sguard{\gamma_0}<|\gamma_0|<3c +|V|\cdot (m'+2c)$ times.
    The resulting path $\gamma_1=\pi_0L_0^{x_1}\pi_1L_1^{y_1}\pi_2$
    is then surely witness, and iterates the loops
    not more than $x_1=3c +|V|\cdot (m'+2c)$
    and $y_1=m'+5c +|V|\cdot (m'+2c)$ times.
    \qed

\end{proof}
%


%
%
\end{document}